\documentclass{article}
\usepackage[final]{neurips_2019}

\usepackage[utf8]{inputenc}
\usepackage[english]{babel}
\usepackage{graphicx}
\usepackage{times} 
\usepackage{amsmath} 
\usepackage{amssymb}
\usepackage{amsthm}
\usepackage{mathrsfs}

\usepackage{bm} 

\usepackage{xcolor}
\usepackage{float}

\newcommand{\NN}{{\mathcal{N}}}
\def\SS{\mathcal S}
\def\OO{\mathcal O}

\DeclareMathOperator*{\col}{col}
\DeclareMathOperator{\rank}{rank}

\newcommand{\commentout}[1]{}

\newtheorem{proposition}{Proposition}
\newtheorem{result}{Result}

\newtheorem{definition}{Definition}
\newtheorem{assumption}{Assumption}
\newtheorem{remark}{Remark} 

\usepackage{tikz}
\usepackage{textcomp}
\usepackage{hyperref}

\newcommand\copyrighttext{%
  \footnotesize \textcopyright 2020 IEEE. Personal use of this material is permitted.
  Permission from IEEE must be obtained for all other uses, in any current or future
  media, including reprinting/republishing this material for advertising or promotional
  purposes, creating new collective works, for resale or redistribution to servers or
  lists, or reuse of any copyrighted component of this work in other works.
}
\newcommand\copyrightheader{%
    \footnotesize \centering Published in 2020 59th IEEE Conference on Decision and Control (CDC) \\
    (DOI: \href{https://doi.org/10.1109/CDC42340.2020.9304112}{10.1109/CDC42340.2020.9304112})%
}
\newcommand\copyrightnotice{%
\begin{tikzpicture}[remember picture,overlay]
    \node[anchor=north,yshift=-10pt] at (current page.north) {\fbox{\parbox{\dimexpr\textwidth-\fboxsep-\fboxrule\relax}{\copyrightheader}}};
    \node[anchor=south,yshift=10pt] at (current page.south) {\fbox{\parbox{\dimexpr\textwidth-\fboxsep-\fboxrule\relax}{\copyrighttext}}};
\end{tikzpicture}%
}

\title{Towards Distributed Accommodation of Covert Attacks in Interconnected Systems
}

\author{Angelo Barboni\\
Department of Electrical and Electronic Engineering, \\
Imperial College London \\
\texttt{a.barboni16@imperial.ac.uk}
\And Thomas Parisini \\
Department of Electrical and Electronic Engineering, \\
Imperial College London, UK, and\\
KIOS Research and Innovation Centre of Excellence, \\
University of Cyprus, and \\
Department of Engineering and Architecture,
University of Trieste, Italy \\
\texttt{t.parisini@gmail.com}
}

\begin{document}

\maketitle
\copyrightnotice

\begin{abstract}
    The problem of mitigating maliciously injected signals in interconnected systems is dealt with in this paper.
    We consider the class of covert attacks, as they are stealthy and cannot be detected by conventional means in centralized settings. 
    Distributed architectures can be leveraged for revealing such stealthy attacks by exploiting communication and local model knowledge.
    We show how such detection schemes can be improved to estimate the action of an attacker and we propose an accommodation scheme  in order to mitigate or neutralize abnormal behavior of a system under attack.
\end{abstract}

\section{Introduction}

Many systems of critical importance consist nowadays of tightly integrated physical and computational components, which may perform control and safety-critical tasks with high reliability requirements.
Additionally, such systems are often composed of several physically interconnected subsystems that exchange information over a network for a number of reasons, ranging from data analysis to design convenience, or simply because the physical system is itself geographically spread over a large area.
As a downside, however, these systems are potentially vulnerable to cyber-attacks which may entail tangible consequences on the physical layer, if not disruption of the system itself.
As observed recently~\cite{lee2016,sobczak2019dos}, attacks constitute a realistic threat, and being able to detect and counteract them to preserve some level of functionality is then of great importance.
In fact, this problem has attracted the interest of the control community over the last decade; see for instance \cite{cheng2017guest} and \cite{Dibaji2019} for recent surveys. 
However, in the majority of cases, the centralized scenario is considered, with only a few works tackling the issue from a distributed perspective~\cite{anguluri2019centralized,gallo2020distributed}. 

Compared to other types of attacks, covert attacks are a class of particularly dangerous attacks, which are undetectable by design in the centralized case~\cite{smith2015covert}. 
In~\cite{barboni19distributed}, it was shown that a specific residual generation scheme allows to detect such attacks, while they remain stealthy within the attacked subsystem.  
The distributed detection strategy is inspired by model-based fault detection (see~\cite{shames2011distributed} and~\cite{boem2017distributedTAC} for instance), with a novel design that accounts for the stealthiness of the attacks (which is not an issue for faults in general). 

In this paper, we extend the detection architecture in \cite{barboni19distributed} with the objective of neutralizing (or at least mitigating) the attacker's malicious effect on the system, that is to say we aim to design a control law that steers the system as close as possible to the desired equilibrium regardless of how the attacker manipulates the control actions (input injection).
To the best of the authors' knowledge, this is the first time that an active countermeasure methodology is proposed in the area of control security, 
and even more in relation with distributed systems.
However, accommodation of faults has received considerable attention from the control community (see~\cite{zhang2008review} for a comprehensive review on the topic).
One way to accomplish this relies on fault estimation in order to cancel their unwanted effect on the system's dynamics via a suitable change of the control action~\cite{polycarpou1995learning,blanke2006diagnosis}. 
This is effective in case of actuator or matching faults, and since input-injection attacks satisfy these same conditions, we gather from this idea in order to compensate the attacks in an additive way.
In dealing with this task from a distributed perspective, a number of issues may arise in case of sparse interconnections between subsystems.
This fundamentally ties with the problem of (partial) input reconstruction~\cite{bejarano2011partial}, as discussed later in the paper.
Additionally, results are hereby presented in discrete time, as opposed to the continuous-time case of~\cite{barboni19distributed}.

The paper is structured as follows: in Section~\ref{sec:problem}, the problem is formulated and the attack model is presented; Section~\ref{sec:detection} provides a short recap of useful results and equations for the detection strategy; in Section~\ref{sec:acco}, the accommodation strategy is presented. 
Finally, in Section~\ref{sec:sim} an academic example is given to show the effectiveness of the proposed distributed accommodation strategy.

\subsection{Notation}
For an ordered index set $\mathcal I$, and a family of matrices $\{M_i \in \mathbb{R}^{n \times m}, i \in \mathcal I\}$, $\col_{i \in \mathcal I}(M_i)$ denotes the vertical concatenation of said matrices.
For brevity, if $x(k)$ is the value of a vector signal $x$ at time $k$, $x^+ \doteq x(k+1)$ denotes the value at the next time step.
Similarly, $x^- \doteq x(k-1)$ denotes the value at the previous time step. 
For a vector $v \in \mathbb{R}^n,\, v_{[m]}$ denotes its $m$-th component.
For a matrix $M \in \mathbb{R}^{n\times m}$, $M^\dagger$ is its pseudo-inverse.
Let $\mathscr R \subset \mathscr X$ be vector spaces, then $\mathscr X / \mathscr R$ denotes the quotient space of $\mathscr X$ by $\mathscr R$.

\section{Problem Statement}\label{sec:problem}

\begin{figure*}
    \centering
    \includegraphics[width=0.7\linewidth]{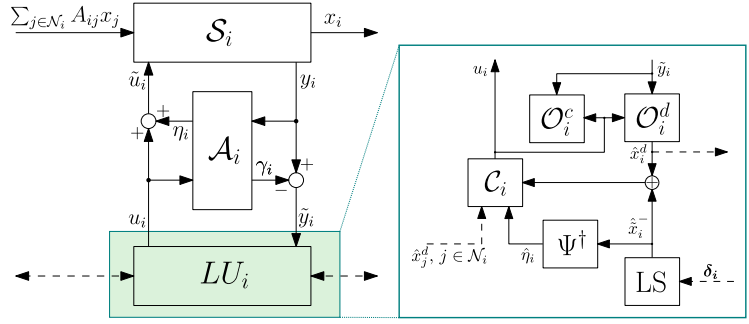}
    \caption{Diagram of a single subsystem's architecture. 
    Details about the accommodation architecture, with attack detection and accommodation measures.}
    \label{fig:arch}
\end{figure*}

We consider a linear time-invariant (LTI) system that can be partitioned into $N$ subsystems, each denoted as $\SS_i$, ${i \in \{1,\dots,N\}}$.
For each subsystem, let $\NN_i$ denote the index set of neighbors of $\SS_i$.
We model each subsystem as a discrete-time LTI system:
\begin{equation}
	\label{eq:planti}
    \mathcal S_i:\left\lbrace
    \begin{aligned}
    	&x_i^+ = A_ix_i + B_i \tilde u_i + \sum_{j \in \mathcal N_i}A_{ij}x_j \\
        &y_i = C_ix_i ,
    \end{aligned}\right.
\end{equation}
where $x_i \in \mathbb{R}^{n_i}$, $\tilde u_i \in \mathbb{R}^{m_i}$, and $y_i \in \mathbb{R}^{p_i}$ are the  local states, inputs, and outputs, respectively. 
$x_j,\, j \in \NN_i$ are the neighbors' states which enter the dynamics of $\SS_i$ through the interconnection matrix $A_{ij} \in \mathbb{R}^{n_i \times n_j}$.

\smallskip
Each subsystem is managed by a local unit $LU_i$ which contains a diagnoser implementing the proposed attack-detection strategy and a given controller $\mathcal C_i$.
Additionally, these units are interconnected along the same topology of the physical interconnections.
As shown on the left-hand side of Fig.~\ref{fig:arch}, the attacker is represented by an interconnection block $\mathcal A_i$, which injects signals $\eta_i$ and $\gamma_i$ in the control and measurement channels, respectively, according to:

\begin{equation*}
	\label{eq:injection}
    \begin{aligned}
        \tilde u_i &= u_i + \eta_i, \\
    	\tilde y_i &= y_i - \gamma_i .
    \end{aligned}
\end{equation*}
Hence, $LU_i$ receives possibly \emph{attacked} measurements $\tilde y_i$, and yields a control action $u_i$ computed accordingly; on the other hand, the plant $\SS_i$ receives the counterfeit control action $\tilde u_i$ to which corresponds an \emph{actual} output response $y_i$. 

The attacker's objective is to remain undetected while steering the system's state to a trajectory different from the nominal one, to its own advantage.
A particular instance of attacks that are stealthy by design are covert attacks, firstly introduced in~\cite{smith2015covert} and investigated in the distributed case in the time domain in~\cite{barboni19distributed}.

\smallskip
\begin{definition}
The attacker $\mathcal A_i$ is covert if the outputs $\tilde y_i$ are indistinguishable from the nominal response $y_i$. $\hfill \triangleleft$
\end{definition}

\smallskip
To perform a covert attack, the attacker implements a model $\tilde\SS_i$ given by

\begin{equation}
    \label{eq:tildeSi}
    \tilde{\mathcal{S}}_i : \left\lbrace
    \begin{aligned}
    	&\tilde x_i^+ = \tilde A_i\tilde x_i + \tilde B_i \eta_i \\
        &\gamma_i = \tilde C_i \tilde x_i ,
    \end{aligned}\right. 
\end{equation}
which is used to compute a ``canceling'' signal $\gamma_i$, for a prescribed input injection $\eta_i$. 
Let $k_a$ be the attack onset instant and let the following assumption hold.
\smallskip
\begin{assumption}\label{ass:knowledge}
The attacker has \emph{perfect knowledge} of the subsystem model, i.e. $(\tilde A_i, \tilde B_i, \tilde C_i) = (A_i,B_i,C_i)$.
$\hfill \triangleleft$
\end{assumption}

\smallskip
It is shown in \cite[Proposition 1]{barboni19distributed} that under Assumption~\ref{ass:knowledge}, if $\tilde x_i(k_a) = 0$, the attacker $\mathcal A_i$ is covert at all times.
Therefore, any residual generator exploiting only $u_i$ and $\tilde y_i$ cannot be used for detection.

We point out that Assumption~\ref{ass:knowledge}, while being difficult to achieve, represents the worst-case scenario of an omniscient attacker who is perfectly stealthy. 
Such a condition frames the detection problem as the most difficult, hence the derived results will also hold for easier cases.

For the scope of the present work, the problem is not limited to attack detection -- which has been already covered in the referenced works -- but rather it focuses on the design of a control input $u_i$ which attenuates the attack's effects.
To the best of the authors' knowledge, this is the first time that a solution to this problem in this distributed flavor is proposed.
Due to the early stage nature of this branch of research, we consider the ideal case in order to obtain basic conditions under which the proposed strategy is effective, and ignore hereby other issues such as robustness to noise.

\section{Detection Architecture}\label{sec:detection}
For sake of completeness, we recap the detection architecture presented\footnote{Extended results including disturbances and detection bounds have been presented in a journal article currently under review.} in~\cite{barboni19distributed}, as well as some important results that are needed for presentation.

At a glance, the detection architecture comprises two observers -- $\OO_i^d$ and $\OO_i^c$ -- and an alarm mechanism that compares a specially constructed residual to a threshold. 
Observer $\OO^d_i$ is decentralized and computes an estimate $\hat x^d_i$ of $x_i$ insensitive to the neighboring states $x_j,\, j \in \NN_i$; $\OO^c_i$ instead is distributed and accounts for the neighboring coupling by including communicated estimates $\hat x^d_j,\, j \in \NN_i$.

Let us define the estimation errors, where we distinguish between the \emph{actual} errors and the \emph{received} (or \emph{attacked}) ones as follows:
\begin{equation}\label{eq:errors}
\begin{aligned}
    &\epsilon_i^d\doteq x_i - \hat x_i^d, &\quad \tilde \epsilon_i^d\doteq x_i - \tilde{x}_i - \hat x_i^d,\\
    &\epsilon_i^c\doteq x_i - \hat x_i^c, &\quad \tilde \epsilon_i^c \doteq x_i - \tilde{x}_i - \hat x_i^c.
\end{aligned}
\end{equation}
Note that the \emph{attacked} errors are in fact those that a diagnoser can compute, as $\tilde y_i$ are the available measurements.
In fact, the actual errors are not available in any way, but they are still useful for analysis and to quantify the attacker's impact.

By construction, the signal $\tilde r^c_i = C_i \tilde \epsilon_i^c$ is sensitive to attacks in neighboring systems.
Suppose that Subsystem $i$ is under attack, the detection logic is as follows.
\begin{itemize}
    \item Each $LU_j, \, j \in \NN_i$ computes $\tilde r_i^c$ and compares it to a threshold $\bar r_i$.
    An alarm $\delta_j \neq 0$ is raised if $\|\tilde r_i^c\| > \bar r_i$.
    
    \item $\delta_i$ is broadcast to each neighbor; conversely, $\SS_i$ receives $\delta_j$ from all its neighbors.
    
    \item If $\delta_j \neq 0, \forall j \in \NN_i$, then $LU_i$ decides that $\SS_i$ is under attack.
\end{itemize}

\smallskip
\begin{remark}
    In~\cite{barboni19distributed}, the alarm variable $\delta_i$ was binary. 
    In the present work, instead, $\delta_i \in \mathbb{R}^{n_i}$, as it will take part into the attack accommodation algorithm, as shown in Section~\ref{sec:acco}.$\hfill\triangleleft$
\end{remark}

\subsection{Decentralized Observer $\OO^d$}
The decentralized observer consists of a discrete-time Unknown Input Observer (UIO):
\begin{equation}
	\label{eq:UIO}
    \mathcal O^d : \left\lbrace
    \begin{aligned}
    	z_i^+  &= F_i z_i + T_i B_i u_i + (K_i^{(1)} + K_i^{(2)}) \tilde y_i \\
        \hat x_i^d &= z_i + H_i \tilde y_i .
    \end{aligned}\right.
\end{equation}
If the existence conditions in~\cite{Chen1996uio} hold, \eqref{eq:UIO} is an UIO for subsystem \eqref{eq:planti}. 
As such, the error dynamics is described by:
\begin{equation}
    \label{eq:UIO_error}
    \epsilon_i^{d+} = F_i \epsilon_i^d.
\end{equation}

\begin{result}[\cite{barboni19distributed}]\label{res:UIO_error}
Let $S_i$ be under attack, if the UIO conditions and Assumption \ref{ass:knowledge} are satisfied, the following equations hold.
The error dynamic of the observer \eqref{eq:UIO} is
\begin{equation}
	\label{eq:UIO_error_att}
    \epsilon_i^{d+} = F_i \epsilon_i^d + (A_i - F_i)\tilde x_i + B_i\eta_i  ,
\end{equation}
while the attacked estimation error defined in \eqref{eq:errors} is given by 
\begin{equation}
	\label{eq:UIO_tildeerr}
    \tilde \epsilon_i^{d+} = F_i\tilde \epsilon_i^d\,.
\end{equation}
$\hfill\square$ 
\end{result}

\smallskip
A consequence of \eqref{eq:UIO_tildeerr} is that the estimate does not converge to the actual state of the system, but rather to the difference $x_i - \tilde x_i$, as can be seen from~\eqref{eq:errors}.

\subsection{Distributed Observer $\OO^c$}

The distributed observer $\OO^c_i$ relies on \emph{decentralized} estimates received over a communication network, and its dynamics is defined as:

\begin{equation}
    \label{eq:Oc}
    \begin{split}
        \mathcal O^c : \hat{x}^{c+}_i &= A_i\hat{x}^c_i + B_i u_i  + L_i (\tilde y_i - C_i\hat{x}^c_i) \\
        &\qquad + \sum_{j\in \mathcal{N}_i} A_{ij} \hat{x}^d_j.
    \end{split}
\end{equation}

Let $F_i^c = A_i - L_iC_i$, if both $F_i$ and $F_i^c$ are stable, then the estimate $\hat x_i^c$ converges to the subsystem's state in attack-free conditions.
In particular, the following result holds.

\smallskip
\begin{result}[\cite{barboni19distributed}]
\label{res:coupled_error}
   Let $S_i$ be under attack, if Assumption~\ref{ass:knowledge} is satisfied.
    The error dynamics of the observer \eqref{eq:Oc} is
        \begin{equation*}
        \label{eq:couplederr}
        \epsilon_i^{c+} = (A_i - L_i C_i)\epsilon_i^c + B_i\eta_i + L_i\gamma_i + \sum_{j \in \mathcal N_i}A_{ij}\epsilon_j^d,
    \end{equation*}
    where $\epsilon_j^d$ is given by \eqref{eq:UIO_error}, \eqref{eq:UIO_error_att}.
    Conversely, the attacked estimation error is given by 
    \begin{equation}
        \label{eq:couplederr_att}
        \tilde \epsilon_i^{c+} = (A_i - L_i C_i)\tilde \epsilon_i^c + \sum_{j \in \mathcal N_i}A_{ij}\epsilon_j^d.
    \end{equation}
    $\hfill\square$
\end{result}

\smallskip
Eq. \eqref{eq:couplederr_att} holds also when the system is not under attack.
It can be seen that the received error (and hence the residual) depends on the \emph{actual} neighbors' decentralized errors $\epsilon^d_j$.
Since under attack $\epsilon_j^d$ evolves according to~\eqref{eq:UIO_error_att}, it follows that under reachability conditions of the pairs $(F_i^c, A_{ij})$, for all $i \in \NN_j$, the error $\tilde \epsilon_i^c$ does not converge to $0$. 

\section{Attack Accommodation}\label{sec:acco}
This section is devoted to the design of a control action $u_i$, which compensates for the effect of the attacker in an attacked subsystem $\SS_i$.
This strategy is triggered after the attack has been successfully detected and isolated.
The accommodation strategy is based on the following observations:
\begin{itemize}
    \item Since the estimation errors converge in nominal conditions, the error $\tilde \epsilon^c_j,\, j \in \NN_i$ can be used to define
    \begin{equation}
        \label{eq:diff}
        d_j \doteq \sum_{l \in \mathcal N_j}A_{jl}\epsilon_l^d = \tilde \epsilon^{c+}_j - F_j^c\tilde \epsilon^c_j,
    \end{equation}
    which we have written in forward form for the sake of convenience.
    The variable $d_j$ represents the aggregate \emph{actual} error of neighboring systems. 
    \item From \eqref{eq:errors} we have that:
    \begin{equation*}
        \epsilon_i^d = \tilde \epsilon_i^d + \tilde x_i.
    \end{equation*}
    Given that $\tilde \epsilon^d_i$ converges to $0$, the actual error $\epsilon_i^d$ depends directly on the state of the attacker's state $\tilde x_i$, which can be seen as superimposed to the nominal system state $x_i^n$.
    Indeed, with initial condition $\tilde x_i(k_a) = 0$, the dynamics of $\SS_i$ can be decomposed as 
    \begin{gather*}\label{eq:state_decomposition}
    \left\lbrace
    \begin{aligned}
        &(x_i^n + \tilde x_i)^+ = A_i(x_i^n + \tilde x_i) + B_i\tilde u_i + \sum_{j \in \mathcal N_i}A_{ij}x_j, \\
        &y_i = C_i x^n_i, \\
        &\gamma_i = C_i \tilde x_i.
    \end{aligned}\right.
    \end{gather*}
\end{itemize}

In view of these observations, the developed strategy aims at constructing an estimate of the attacker's state using neighboring errors.
For this purpose, we redefine the alarm signal introduced in Section~\ref{sec:detection} using \eqref{eq:diff} as:
\begin{equation*}
    \delta_j = \left\lbrace 
    \begin{aligned}
    & d^-_j \quad&&\text{if }\|\tilde \epsilon_j^c\| > \theta_j, \\
    & 0 &&\text{otherwise},
    \end{aligned}\right.
\end{equation*}
for some suitable threshold $\theta_j$. Suppose that $\SS_i$ is under attack; if $d_j \neq 0,\, \forall j \in \NN_i$, then the attack is detected in the same way as previously presented.

A one-step lag estimate of $\hat{\tilde{x}}_i^-$ of $\tilde x_i$ can be obtained by solving the following Least Squares (LS) problem:
\begin{equation}
    \label{eq:LS}
    \hat{\tilde x}_i^- = \arg\min_{\xi} \left\lbrace \frac{1}{N_i}\sum_{j \in \NN_i}\left\| \delta_j - A_{ji}\xi \right\|_2^2 \right\rbrace .
\end{equation}
The solution to \eqref{eq:LS} is given by:
\begin{equation*}
    \label{eq:LSsol}
    \hat{\tilde x}_i^- = \left( \sum_{j \in \NN_i} A_{ji}^\top A_{ji} \right)^{-1} \left( \sum_{j \in \NN_i} A_{ji}^\top \delta_j \right).
\end{equation*}

\smallskip
\begin{remark}
With respect to just performing attack detection, the designed accommodation strategy requires that the local diagnoser also knows the \textit{outbound} interconnection matrices $A_{ji}$.$\hfill\triangleleft$
\end{remark}

\smallskip
\begin{remark}
Let $\bm{\delta}_i \doteq \col_{j \in \NN_i}(\delta_j)$ and $\bm{A}_i \doteq \col_{j \in \NN_i} (A_{ji})$.
Then problem \eqref{eq:LS} can be equivalently reformulated in matrix form as:
\begin{equation*}
    \label{eq:LS2}
    \hat{\tilde x}_i^- = \arg\min_{\xi} \left\| \bm{\delta}_i  - \bm{A}_i  \xi \right\|_2^2,
\end{equation*}
for which standard solution techniques can be applied.
$\hfill\triangleleft$
\end{remark}

\smallskip
Uniqueness of the estimate depends on well-known rank conditions on $\bm{A}_i$, which we refer to as the \emph{aggregate interconnection matrix}.
We defer discussion to the respective subsections.

To further develop our analysis, consider \eqref{eq:planti} and a control action of the form 

\begin{equation}\label{eq:control}
    u_i = K_i(\hat x^d_i + \hat{\tilde x}_i^-) + \sum_{j \in \NN_i}K_{ij}\hat x_j^d - \hat\eta_i .
\end{equation}
Matrices $K_{ij}$ can be optimally chosen~\cite{vsiljak1978large} to minimize the effects of neighbors on the local dynamics.
Without loss of generality, since our analysis focuses on attack compensation on $\SS_i$, we can assume exact ``cancellation'' of neighboring states.
Let $\epsilon^a_i \doteq \tilde x_i - \hat{\tilde x}_i^-$ and $\epsilon_i^\eta \doteq \eta_i - \hat\eta_i$, we have that:

\begin{align}\label{eq:exp_dynamics}
    \begin{aligned}
        x_i^+ &= A_ix_i + B_iK_i(\hat x^d_i + \hat{\tilde x}_i^-) + B_i(\eta_i - \hat\eta_i)\\
              &\hspace{4ex} +\sum_{j \in \NN_i}(A_{ij}+B_iK_{ij})x_j - \sum_{j \in \NN_i}B_iK_{ij}\epsilon_j^d     \\
              &= A_ix_i + B_iK_i(x_i - \tilde\epsilon_i^d - \tilde x_i + \hat{\tilde x}_i^-) + B_i\epsilon^\eta_i    \\
              &\hspace{4ex} +\sum_{j \in \NN_i}(A_{ij}+B_iK_{ij})x_j - \sum_{j \in \NN_i}B_iK_{ij}\epsilon_j^d     \\
              &= (A_i + B_iK_i)x_i + \sum_{j \in \NN_i}(A_{ij}+B_iK_{ij})x_j +  B_i\epsilon^\eta_i  \\
              &\hspace{4ex} - B_iK_i\epsilon_i^a - B_iK_i\tilde\epsilon_i^d - \sum_{j \in \NN_i}B_iK_{ij}\epsilon_j^d.
    \end{aligned}
    \raisetag{19ex}
\end{align}

In \eqref{eq:exp_dynamics}, $K_i$ can be designed to achieve asymptotic closed-loop stability, and, with proper design of the UIOs, the respective error terms are asymptotically vanishing. 
In this case, the attacked subsystem is still driven not only by the attacker's internal state, but also by the injected input $\eta_i$.
This entails the necessity of reconstructing such an input from the estimate $\hat{\tilde x}_i^-$ that has been computed.

\smallskip
\begin{remark}
Although no assumptions are made on the controller, \eqref{eq:control} is presented to simplify the analysis in the case of linear control.
In fact, the proposed accommodation strategy consists of an additive input $\hat\eta_i$ and an additive estimate compensation $\hat{\tilde x}_i$, which can in principle be used in several control designs.
$\hfill\triangleleft$
\end{remark}

\subsection{Full-rank Aggregate Interconnection}
If $\rank \bm{A}_i = n_i$, then the solution to~\eqref{eq:LS} is unique, and we have that 
$$
    \hat{\tilde x}_i^- = \tilde x_i^-.
$$
As a result, the error term $\epsilon^a_i$ obeys the dynamics
\begin{equation*}
    \label{eq:fr:tildex_error}
    \epsilon_i^{a+} = \tilde A_i \epsilon_i^a + \tilde B_i \Delta\eta_i,
\end{equation*}
where $\Delta\eta_i = \eta_i - \eta_i^-$.
It is possible to obtain an estimate $\hat \eta_i$ which solves the input reconstruction problem relative to the dynamics \eqref{eq:tildeSi} for a known $\tilde x_i^-$.

\smallskip
\begin{definition}[Relative degree]\label{def:reldegree}
Consider a linear system of the form \eqref{eq:planti}. 
Let $c_l$ be the $l$-th row of matrix $C_i$; if there exist integers $r_l$ such that
$$
C_iA_i^kB_i = 0, \; C_iA_i^{r_l-1}B_i \neq 0,\; \forall k < r_l - 1
$$
and
$$
 \rank \begin{bmatrix} c_1 A_i^{r_1 - 1} B_i \\
                       \vdots \\
                       c_p A_i^{r_{p_i} - 1} B_i\end{bmatrix} = m_i
$$
then $r = [r_1, \dots, r_{p_i}]$ is called the relative degree of system~\eqref{eq:planti}.$\hfill\triangleleft$
\end{definition}

\smallskip
Let us define the stacked vector of the attacker's state and inject input estimates
\begin{equation*}\label{eq:stack_tildex}
    \hat{\tilde{\bm{x}}}_i \doteq \col_{t \in \{1,\dots,r\}} (\hat{\tilde x}_i(k-t)),\; 
    \hat{\bm{\eta}}_i \doteq \col_{t \in \{1,\dots,r\}} (\hat{\eta}_i(k-t-1)),
\end{equation*}
respectively, and the input-to-state dynamic matrix $\Psi_i$ \cite{edelmayer2004input}:
\begin{equation*}\label{eq:Psi}
    \Psi_i \doteq \begin{bmatrix}
        B_i          & 0            & \dots & 0 \\
        A_i B_i      & B_i          & \dots & 0 \\
        \vdots       & \vdots       & \ddots& \vdots\\
        A_i^{r-1}B_i & A_i^{r-2}B_i & \dots & B_i
    \end{bmatrix},   
\end{equation*}
where $r \geq n_i$.
Then, we can state the following.

\begin{proposition}\label{prop:fr:inputest}
    The injected signal $\eta_i(k)$ can be estimated in finite time if and only if system~\eqref{eq:tildeSi} is left-invertible. 
    Furthermore,
    \begin{equation}\label{eq:fr:etaestimate_vec}
        \hat{\bm{\eta}}_i = \Psi_i^\dagger\hat{\tilde{\bm{x}}}_i
    \end{equation}
    and
    \begin{equation}\label{eq:fr:etaestimate}
        \hat\eta_i(k) = \eta_i(k - r_0 - 1) = \hat{\bm{\eta}}_{i[r]}
    \end{equation}
    $\hfill\square$
\end{proposition}
\begin{proof}
If system \eqref{eq:tildeSi} is left-invertible, then ${\hat{\tilde{\bm{x}}}_i = \Psi_i\hat{\bm{\eta}}_i}$ admits a unique solution, given by~\eqref{eq:fr:etaestimate_vec}.
The delay in the input estimate~\eqref{eq:fr:etaestimate} follows by Definition~\ref{def:reldegree}, with $c_i$ taken as the canonical euclidean basis vectors.
In this, $r_0$ is the largest component of the vector relative degree, i.e. $$r_0 = \max_{l \in \{ 1,\dots,r_{p_i}\}}(r_l).$$
The additional lag step is given by the intrinsic delay in the estimate $\hat{\tilde x}_i$.
\end{proof}

\smallskip
\begin{remark}
The left-invertibility condition necessary to obtain $\hat{\bm{\eta}}_i$ is implied by existence conditions for the UIO $\OO_i^d$~\cite{Hou1998input}, hence no further assumptions are made on the problem setting.
$\hfill\triangleleft$
\end{remark}

\smallskip
In this case, from an analytical point of view, the compensation mismatch depends on differences of the attacker's input signal because of the intrinsic delay of the estimation procedure.
For constant injected inputs, e.g. steady offsets, it follows that attack compensation is exact.

\smallskip
\begin{remark}
Notice that ${\rank A_{ji} = n_i}$, ${\forall j \in \NN_i} \Rightarrow {\rank \bm{A}_i = n_i}$, but the converse is not true in general.$\hfill\triangleleft$
\end{remark}

\subsection{Low-rank Aggregate Interconnection}
In this subsection, we consider the case ${\rank \bm{A}_i < n_i}$.
This is attained when a subset of the components of $x_i$ does not influence the neighbors.
More formally, the aggregate interconnection is low-rank if $\exists\, g_i \in \mathbb{N}$ such that
$$
 \dim \left( \bigcap_{j \in \NN_i} \ker A_{ji} \right) = g_i > 0.
$$

We can introduce a decomposition of the state space $\mathbb{R}^{n_i}$ into a \emph{non-interacting} subspace $X_i^{\perp} \doteq \ker \bm{A}_i$ and an \emph{interacting} one $X_i^{||} \doteq \mathbb{R}^{n_i} / X_i^{\perp}$.
Clearly, we have that $\dim X^{||} = n_i - g_i$. 
Consequently, we can define respective canonical projections of $\mathbb{R}^{n_i}$ onto these subspaces, and in particular we consider $P_i: \mathbb{R}^{n_i} \to X_i^{||}$.

With this projection, the solution of the LS problem will be exact only on the interacting subspace.
In particular, we have that 
\begin{align*}
    \begin{aligned}
        & \hat{\tilde x}_i^{||-} \doteq P_i \hat{\tilde x}_i^- = P_i \tilde x_i^-, \\
        &  \hat{\tilde x}_i^{\perp -} \doteq (I - P_i) \hat{\tilde x}_i^- = 0.
    \end{aligned}
\end{align*}
By means of this projection, it is possible to reframe the problem as the input reconstruction for the system
\begin{equation}\label{eq:proj_sys}
    \tilde \SS^{||}_i:\left\lbrace
    \begin{aligned}
        &\tilde x_i^+ = \tilde A_i \tilde x_i + \tilde B_i \eta_i \\
        &\tilde x_i^{||} = P_i\tilde x_i,
    \end{aligned}\right.
\end{equation}
where the second equation is considered as the system output. 
The input-output matrix can be rewritten as
\begin{equation*}
    \label{eq:Psipar}
    \Psi^{||}_i \doteq \begin{bmatrix}
        P_iA_i^{r_0-1} B_i  & 0                   & \dots  & 0 \\
        P_iA_i^{r_0} B_i    & P_iA_i^{r_0-1} B_i  & \dots  & 0 \\         
        \vdots              & \vdots              & \ddots & \vdots\\
        P_iA_i^{r-1}B_i     & P_iA_i^{r-2}B_i     & \dots  & P_iA_i^{r_0-1} B_i 
    \end{bmatrix},
\end{equation*}
where $r \geq n_i$ and $r_0 = \max_{l \in \{ 1,\dots,r_{p_i}\}}(r_l)$ is the maximum relative degree of~\eqref{eq:proj_sys}.

\smallskip
\begin{definition}[\cite{trentelman2012control}]
The constant $\lambda \in \mathbb{C}$ is an $(A,C)$-unobservable eigenvalue if $\rank \begin{bmatrix} \lambda I - A \\ C \end{bmatrix} < n$, with $n = \dim A$. $\hfill\triangleleft$
\end{definition}

\smallskip
\begin{proposition}
    The input $\eta_i(k)$ in \eqref{eq:proj_sys} can be estimated in finite time if and only if the set of invariant zeros of $\tilde \SS_i^{||}$ is identical to the set of $(\tilde A_i, P_i)$-unobservable eigenvalues.
    Furthermore, 
    \begin{equation}\label{eq:etapar_vec}
        \hat{\bm{\eta}}_{i} = \left(\Psi_i^{||}\right)^\dagger \hat{\tilde{\bm{x}}}_i^{||}
    \end{equation}
    and
    \begin{equation}\label{eq:etapar}
        \hat\eta_i(k) = \eta_i(k-r_0-1) = \hat{\bm{\eta}}_{i[r]}
    \end{equation}
     $\hfill\square$
\end{proposition}
\begin{proof}
Using~\cite[Theorem 4.10]{bejarano2009unknown}, we ensure that~\eqref{eq:proj_sys} is left invertible.
In that case,~\eqref{eq:etapar_vec} admits a unique solution.
The remaining considerations follow those in Proposition~\ref{prop:fr:inputest}.
\end{proof}

\smallskip
Finally, reusing $\hat \eta_i$ in \eqref{eq:proj_sys} allows for a forward estimation of $\tilde x_i^\perp$.
In fact, by definition, such part of the state does not \emph{directly} affect the neighboring dynamics and hence the communicated errors.
As a result, the only way of reconstructing the entire $\hat{\tilde x}_i$ is via the state update equation of the local attacker's model~\eqref{eq:tildeSi}.

\smallskip
\begin{remark}
The decomposition method presented in this subsection can be related to~\cite{bejarano2011partial}, where however the problem statement is different. 
Despite the differences in the setting, in the cited work, the input matrix (in the present case, the outbound interconnection) is decomposed onto orthogonal subspaces, and partial left-invertibility conditions are developed for the decomposed system.$\hfill\triangleleft$
\end{remark}

\section{Simulation Example}\label{sec:sim}
We show the effectiveness of the proposed method on a simple numerical example on regulation.
This example is meant to illustrate the practicality of the proposed procedure.

The overall system comprises $N=5$ subsystems and the topology is described the following neighbors sets: $\NN_1 = \{ 2, 3\}$, $\NN_2 = \{ 1, 3, 4\}$, $\NN_3 = \{ 1,2\}$, $\NN_4 = \{ 1,2,5\}$, $\NN_5 = \{4\}$.
We consider the full and low rank interconnection cases, and for both we use the following subsystem's dynamics. $\forall i \in \{1, \dots, N\}$: 
\begin{align*}
    A_i = \begin{bmatrix}
        0.4 & 0.2 \\
        0   & 0.3
    \end{bmatrix} && 
    B_i = \begin{bmatrix}
        0 \\ 1
    \end{bmatrix} && 
    C_i = I && D_i = 0.
\end{align*}
Choice of $A_{ij}$ will be presented in two separate examples in the following subsections.
Individual pairs of observers are designed for the two systems, as presented in Section~\ref{sec:detection}.
Each system implements a control law of the form~\eqref{eq:control}, which optimally decouples the neighboring dynamics and achieves a prescribed rate of convergence.

For simplicity and without loss of generality, all subsystems implement the same dynamics and the same interconnection.
At time $k_a = 20$, $\SS_3$ is covertly attacked according to the model presented in Section~\ref{sec:problem}.
The attacker's objective is to introduce a steady-state error into the regulator.

\subsection{Full Rank Interconnections}
For this case, the interconnection matrix is chosen as
$$
A_{ij} = \begin{bmatrix} 0.1 & 0 \\ 0 & -0.01 \end{bmatrix}.
$$
Results for this example can be seen in Fig.~\ref{fig:fr:compensation} and~\ref{fig:fr:estimates}.
In the former, components of the true state are shown, and it is possible to see that the effect of the attack is compensated according to the controller's dynamics with a certain delay.
This delay is particularly evident in Fig.~\ref{fig:fr:estimates}: in the left-hand side figure, the one-step delay intrinsic in the computation of $\hat{\tilde x}_3$ is shown.
This translates into an $r_0 + 1$ delay for the computation of $\hat{\eta}_3$, as depicted on the right-hand side of Fig.~\ref{fig:fr:compensation}.

\begin{figure}
    \centering
    \includegraphics[width=.8\linewidth]{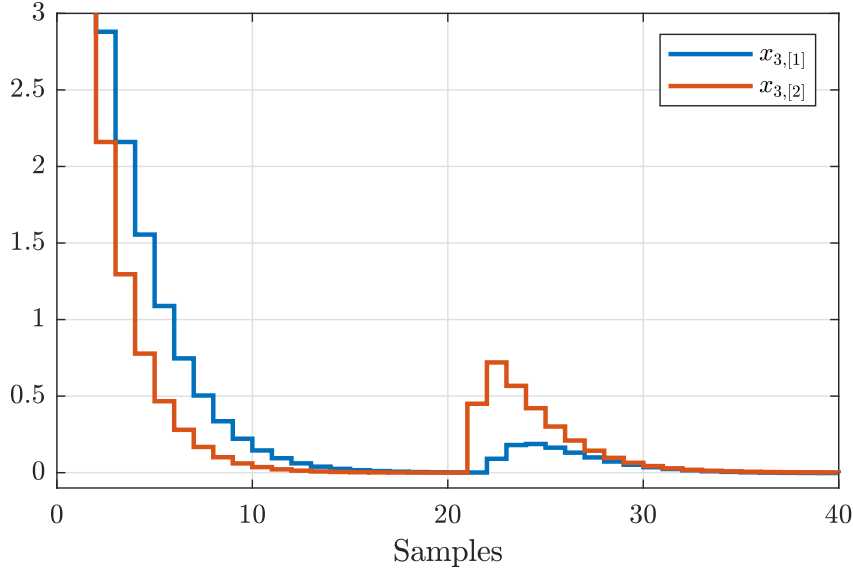}
    \caption{State trajectories of $\SS_3$ (full rank interconnection).}
    \label{fig:fr:compensation}
\end{figure}

\begin{figure}
    \centering
    \includegraphics[width=\linewidth]{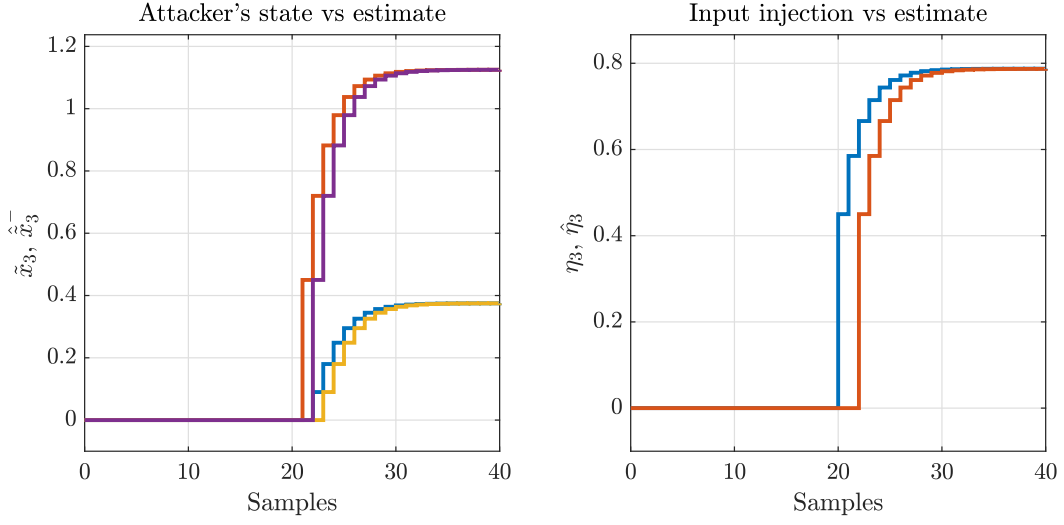}
    \caption{(left) Attacker's state trajectories and LS estimates; (right) Actual injected signal $\eta_3$ and reconstructed estimate $\hat\eta_3$ (full rank interconnection).}
    \label{fig:fr:estimates}
\end{figure}

\subsection{Low Rank Interconnections}
In this case, the interconnection matrix is chosen as
$$
A_{ij} = \begin{bmatrix} 0.1 & 0 \\ -0.1 & 0\end{bmatrix}.
$$
As in the previous subsection, the results can be seen in Fig.~\ref{fig:lr:compensation} and~\ref{fig:lr:estimates}.
It can be noticed how the performance of the accommodated system is not particularly different from that obtained in the full rank case.
On the other hand, the estimation of the attacker's state $\tilde x_3$ is affected by a longer delay, as shown in the left-hand side of Fig.~\ref{fig:lr:estimates}. 
These estimates are obtained using the reconstructed input as computed by~\eqref{eq:etapar}, and which are shown in the right-hand side of the picture.

\begin{figure}
    \centering
    \includegraphics[width=.8\linewidth]{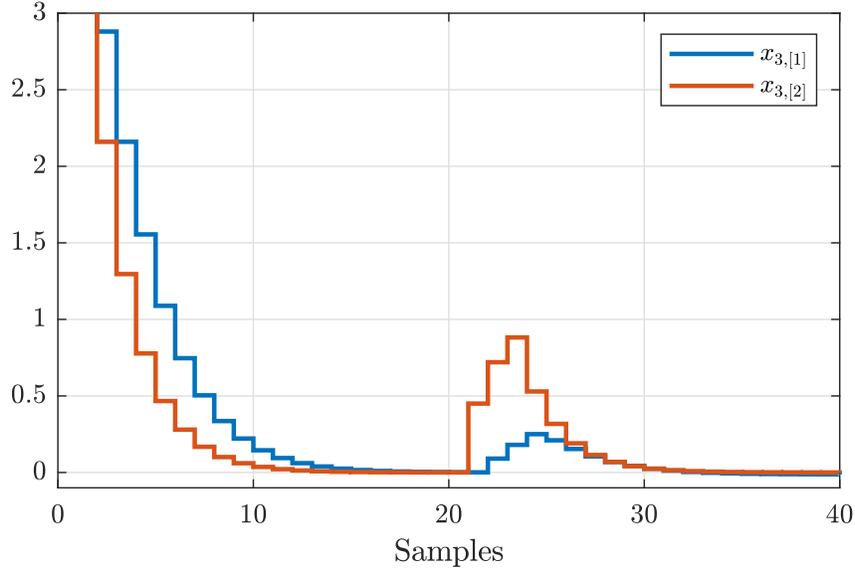}
    \caption{State trajectories of $\SS_3$ (low rank interconnection).}
    \label{fig:lr:compensation}
\end{figure}

\begin{figure}
    \centering
    \includegraphics[width=\linewidth]{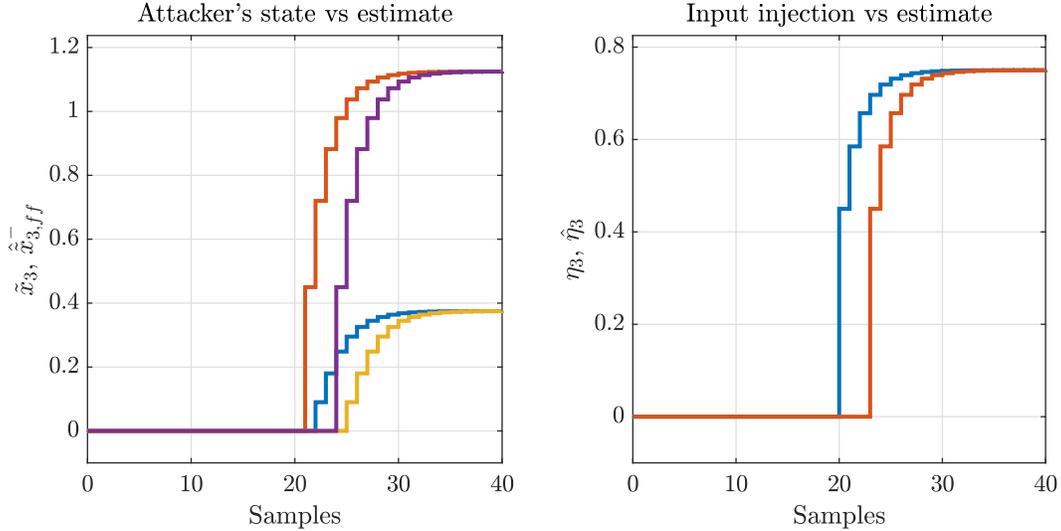}
    \caption{(left) Attacker's state trajectories and LS estimates; (right) Actual injected signal $\eta_3$ and reconstructed estimate $\hat\eta_3$. Notice the longer delay needed due to the projection procedure (low rank interconnection).
    }
    \label{fig:lr:estimates}
\end{figure}

\section{Conclusions and Future Work}
In this paper, a novel distributed methodology for accommodation of stealthy local attacks in interconnected systems is presented.
To the best of the authors' knowledge, this is the first time a step is done in this direction, and the approach can be in principle applied to other cases where locally unobservable states have no effects on residuals.
Given the early stage nature of the presented methodology, additional work is being done on characterizing robustness and the impact of noise and disturbances on the estimates.

\subsubsection*{Acknowledgments}
This work has been partially supported by the EPSRC Centre for Doctoral Training in High Performance Embedded and Distributed Systems  (HiPEDS, Grant Reference EP/L016796/1), the European Union’s Horizon 2020 Research and Innovation Programme under grant agreement No. 739551 (KIOS CoE), and by the Italian Ministry for Research in the framework of the 2017 Program for Research Projects of National Interest (PRIN), Grant no. 2017YKXYXJ.


\end{document}